\documentclass{article}

\usepackage{graphicx,amsmath,psfrag} 
\usepackage{algorithm}
\usepackage{standalone}
\usepackage{tikz}
\usepackage{algpseudocode}               

\usepackage{subcaption}

\usepackage[leftcaption]{sidecap}
\usepackage{todonotes}
\usepackage{mathtools}
\usepackage{cleveref}
\usepackage{amsthm}
\usepackage{complexity}
\usepackage[misc]{ifsym}
\usepackage{commands}

\usepackage{xcolor}
\usepackage{amssymb}
\usepackage{standalone}
\usepackage{tikz}
\usetikzlibrary{arrows.meta}
\usetikzlibrary{arrows}

\usepackage{commands}

\usepackage{algpseudocodex}

\newtheorem{theorem}{Theorem}

\newtheorem{proposition}{Proposition}
\newtheorem{lemma}{Lemma}

\title{Minimum Spanning Trees for Square Crop Plots}

\date{}

\author{
Mingyang Gong \footnote{
Gianforte School of Computing, Montana State University, Bozeman, MT 59717, USA. Email: {\tt mgong4@ualberta.ca}} \and
Adiesha Liyanage\footnote{Gianforte School of Computing, Montana State University, Bozeman, MT 59717, USA. Email: {\tt a.liyanaralalage@montana.edu}}
\and
Braeden Sopp\footnote{Gianforte School of Computing, Montana State University, Bozeman, MT 59717, USA. Email: {\tt braeden.sopp@student.montana.edu}}
\and 
Muzhou Chen \footnote{Gianforte School of Computing, Montana State University, Bozeman, MT 59717, USA. Email: {\tt  muzhouchen@outlook.com}}
\and
Binhai Zhu\footnote{Gianforte School of Computing, Montana State University, Bozeman, MT 59717, USA. Email: {\tt bhz@montana.edu}}
}

\begin{document}
\maketitle

\begin{abstract}
Motivated by accessing crop plots in a field, where each crop plot can only be visited by a given pair of entry/exit points (we have three different types, each having four cases), we study the corresponding Minimum Spanning Tree (MST) problem of such a planar set of axis-aligned unit squares. It turns out that this crop plot distance does not satisfy triangle inequality, even though the whole setup of the problem is geometric. Hence additional care must be taken. The main results of this paper are as follows: (1) we prove that computing the MST of a set $P$ of unit squares under the crop plot distance is NP-hard, (2) the MST of $P$ can be approximated with a factor-4 approximation.
\end{abstract}

\section{Introduction}

Imagine in a farm we have many axis-aligned crop plots each having an even or odd number of rows of crops; moreover, the rows in one plot could be either horizontal or vertical. Obviously one could take care of the rows of crops (say getting rid of weeds) by traversing through the lists of separators between the crops. Under this assumption, to process a whole crop plot one could
only enter and exit a plot through two entry/exit points (ends of two extremal separators). For simplicity, we assume all the plots are axis-aligned unit squares. Then the distance traversed within a crop plot is fixed and is at least one. Hence, if there are $n$ squares the total cost
to traverse them internally is $C\geq n$. (At the end of the paper, we would also discuss using rectangles to model these crop plots.)

It turns out that, for each plot, if we order the vertices of the outermost separators in clockwise order as $\langle a,b,c,d\rangle$ (starting with the bottom-right vertex $a$), then the
three types of entry/exit points are: $\{a,c\}$ or $\{b,d\}$ (type-1), $\{a,b\}$ or $\{c,d\}$ (type-2), and, $\{a,d\}$ or $\{b,c\}$ (type-3). See Figure~1 for an illustration.

\begin{figure}[t]
    \centering
    \includegraphics[width=0.5\textwidth]{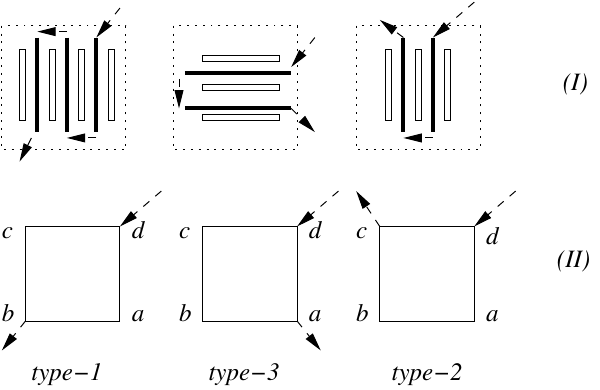}
    \label{fig:hidden-face}
	\caption{Crop plots (I) and their abstraction (II). In (I), skinny rectangles represent crops and thick segments represent separators farmers can walk on. In (II), the three types of (abstracted) squares are illustrated. For the type-1 square, the shown entry/exit points are $d/b$ respectively; but they could be $b/d$, $a/c$ or $c/a$. Other types of squares are similar.}
\end{figure}

Naturally, starting with a square, an interesting problem is to visit all the squares following the above constraints; i.e., we could have the TSP (or the shortest Hamiltonian Path) problem of square plots. This problem is obviously NP-hard for squares (just making the squares very small relative to their inter-distances, then following the classical reduction by Papadimitriou \cite{DBLP:journals/tcs/Papadimitriou77}). A natural approximation for TSP on planar points involves the computation of a minimum spanning tree of these points \cite{Christofides1976}. In the case of square plots, this becomes non-trivial as computing a minimum spanning tree of these square plots is NP-hard.

It turns out that this new crop-plot distance ($d_{cp}(-,-)$) between axis-aligned squares does not satisfy the triangle inequality (see Fig.~3). Then, a more fundamental question arises: if a set $P$ of $n$ square plots are given, what is the complexity of computing a minimum spanning tree of
$P$? Can we obtain a good approximation if it is hard to compute?

Hence, we abstract the problem as follows. Given a set $P$ of $n$ disjoint square plots, how to compute a minimum spanning tree of $P$, ${MST}_{cp}(P)$, where the crop-plot distance between two neighboring
squares in the MST is the shortest (geodesic) distance between their selected entry/exit corner points. (In ${MST}_{cp}(P)$ the entry and exit points of any square can be thought of as connected with some fixed cost. The total such interior cost over all square plots is a constant $C\geq n$.) Hence, the total cost of
${MST}(P)$ is the cost of ${MST}_{cp}(P)$ plus $C$. As a well-known fact, MST could be used to compute an approximate TSP for the geometric or metric-space version where the triangle inequality holds \cite{Christofides1976}.

We summarize the main results of this paper as follows:
\begin{enumerate}
    \item Computing ${MST}_{cp}(P)$ (or ${MST}(P)$) of $P$ is NP-hard, and the reduction is from the Monotone Planar 3SAT (MP-3SAT for short). Different from the related previous proofs, special care must be taken to enforce connectivity.
    \item ${MST}(P)$ can be approximated with a factor of 4.
\end{enumerate}

\noindent
{\bf Previous Related Work~~} The Planar 3-SAT based NP-hardness proof for planar rectilinear problems was initiated by Knuth and Raghunathan \cite{DBLP:journals/siamdm/KnuthR92}. In 2012, de Berg and Khosravi showed that Monotone Planar 3-SAT (MP-3SAT) is NP-hard even when the variable gadget can be drawn on a horizontal line and clauses can be drawn below and above the line (more precisely, the clauses with all positive literals can be drawn above the line while those with all negated literals can be drawn below the line --- all the clauses are rectilinearly connected to the variables) \cite{DBLP:journals/ijcga/BergK12}. This setting makes it easier to simplify some NP-hardness proofs; for instance, the point dispersion problem on rectilinear segments \cite{DBLP:journals/tcs/HeLZZ21}. In fact, we will use a similar reduction from MP-3SAT, but the details are much more involved since additional care must be taken regarding the connectivity. As a result, not all squares can be placed on a integral grid and some extra effort is needed to make sure the residue is small enough so that it does not affect the proof.

In \cite{yang2007minimum}, Yang et al. introduced the Minimum Spanning Tree with Neighborhoods (MSTN) problem—a generalization of geometric MSTs where a tree must span representative points selected from given 2D regions—and analyzed its complexity. They proved that the problem is NP-hard when the neighborhoods consist of disjoint discs and rectangles, and focus on algorithmic solutions for the case of disjoint 2D unit discs.

When the objects are line segments and the objective is to connect them into a Hamiltonian Circuit or Hamiltonian Path, only by adding straight edges connecting the endpoints of the input segments, the problem was shown to be NP-complete by Rappaport \cite{DBLP:journals/siamcomp/Rappaport89}. In this proof the segments could meet at their endpoints, later the problem was shown to be NP-complete even if the input segments are disjoint and the input segments have only 4 orientations \cite{DBLP:journals/dcg/AkitayaKKRST22,DBLP:journals/corr/abs-2108-12812}. Note that these reductions are from the Hamiltonian Path problem in planar cubic graphs (while the initial construction is from Rosenstiehl and Tarjan's rectilinear layout \cite{DBLP:journals/dcg/RosenstiehlT86}, the complete construction is not rectilinear). In our case, the objects are axis-parallel squares, the crop-plot distance is geodesic and our objective is to minimize the traversing cost.

This paper is organized as follows. In Section 2 we give necessary definitions. In Section 3 we present our NP-hardness proof. We present a factor-4 approximation for our MST problem in Section 4. We conclude the paper with some discussions in Section 5. We leave four detailed proofs of Proposition 2 and 4, Lemma 6 and Theorem 8 (which will be marked a *) in the appendix.


\section{Definitions}
Let $P$ be $n$ input square plots, we focus on the complexity of computing ${MST}_{cp}(P)$ in Section 2 and 3. We will come back to ${MST}(P)$ in Section 4 when we discuss its approximation. To simplify presentation, from now on we simply call ${MST}_{cp}(P)$ a spanning tree \(\mathcal{T}\).
Recall that $c({MST}(P))=c({MST}_{cp}(P))+C$, with $C\geq n$ being a fixed number.

Let \(P = \{p_{1}, p_{2}, ..., p_{k}\}\) 
be the set of plots in \(\R^2\), with
each plot being a square with unit width and height,
and let \(t: P \to \{1,2,3\} \) record the type of each plot.
A spanning tree of \(P\) is a tree whose set of vertices is \(P\)
along with an assignment of an entry and exit points to each plot which
agree with its type.
Formally we may represent this as a four tuple 
\(\mathcal{T} = (T, t, \pltin, \out)\) where
\(T\) is a spanning tree on the plots in \(P\), \(t\) is the map recording each plot's type
and \(\pltin: P \to \R^2\) and \(\out: P \to \R^2\) are the functions that
assign the input and output locations of each plot in the plane.
Given a spanning tree \(\mathcal{T}\) and
a \(P' \subseteq P\), we write \(\mathcal{T}[P']\) to denote the spanning forest that corresponds
to the subgraph of \(P'\) induced by \(T\).
In~\cref{fig:plot-notation}, we describe how plots are notated in our figures.
\begin{figure}
    \centering
        \begin{subfigure}[t]{0.2\textwidth}
        \centering
        \includegraphics[width=0.8\textwidth]{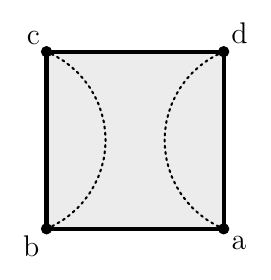}
        \label{fig:type3mod}
    \end{subfigure}
    \hfill 
    \begin{subfigure}[t]{0.2\textwidth}
        \centering
        \includegraphics[width=0.8\textwidth]{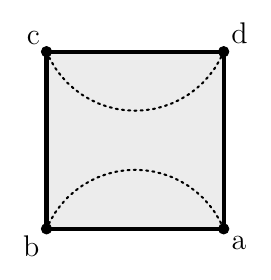}
        \label{fig:type2mod}
    \end{subfigure}
    \hfill
    \begin{subfigure}[t]{0.2\textwidth}
        \centering
        \includegraphics[width=0.8\textwidth]{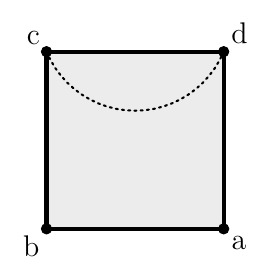}
        \label{fig:type2moddir}
    \end{subfigure}

    \caption{Moving from left to right. The first diagram is used to represent type-3 plots, the second diagram is used to represent type-2 plots. Type-1 plots are not used in this reduction. Finally, the third diagram represent a selection of entrance and exit for a plot. In this case the point \(c\) is chosen as the entrance and \(d\) as the exit for the type-2 plot shown to the right.}
    \label{fig:plot-notation}
\end{figure}

Given two plots \(A, B  \in P\) and a fixed choice of \(\pltin\) and \(\out\) functions,
we define the crop-plot distance as follows:
\begin{equation*}
    d_{cp}(A, B) = \min_{a\in\{\pltin(A), \out(A)\}}\min_{b \in \{\pltin(B), \out(B)\}} d(a, b)
\end{equation*}
where \(d(\pltin(A), \out(B))\) is the length of the shortest path between
\(\pltin(A)\) and \(\out(B)\) that avoid crossing through any plots.
Using the crop-plot distance, we define the cost
of an edge \((A,B) \in T\) as \(d_{cp}(A, B)\). Note the crop distance does not need to obey the triangle inequality as shown in~\cref{fig:crop-triangle}.
The \textit{cost} of a spanning tree \(\mathcal{T}\) is thus defined as
\begin{equation*}
    c(\mathcal{T}) = \sum_{(A,B) \in E(T)}d_{cp}(A, B).
\end{equation*}

\begin{figure}
    \centering
    \includegraphics[width=0.5\linewidth]{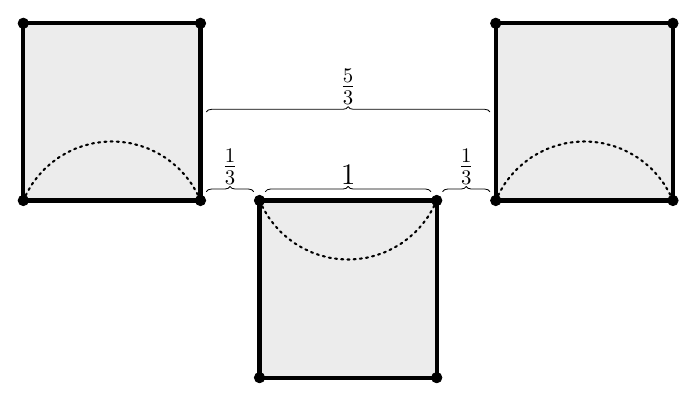}
    \caption{This figure is drawn to scale. Given the assignment of entrances and exits on the plots, the center plot has crop plot distance 1/3 from its neighbors while the neighbors have distance 5/3 from each other.}
    \label{fig:crop-triangle}
\end{figure}

\subsection{Problem Description}
Formally, the decision problem version of the \cpprob
is as follows
\begin{description}
    \item{\textbf{Input}} A collection of disjoint unit squares \(P = \{p_{1}, p_{2}, ..., p_{k}\}\) in \(\R^2\), an assignment of types \(t:P\to \{1,2,3\}\),
    and a nonnegative real number \(\lambda\).
    \item{\textbf{Problem}} 
    Decide if there exists a spanning tree of \(P\), \(\mathcal{T}\), such that
       \[
       c(\mathcal{T}) \leq \lambda.
       \]
\end{description}

\section{Reduction}

We show that the \cpprob is NP-hard using a reduction from 
Monotone Planar 3SAT.
In monotone planar 3SAT each clause contains only positive or
negative literals and the graph admits a rectilinear graph drawing
as seen in~\cref{fig:mprsat-example}.

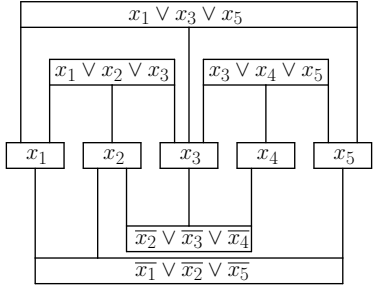
\begin{figure}
    \centering
    \resizebox{.4\linewidth}{!}{\begin{tikzpicture}[line cap=round,line join=round,>=triangle 45,x=1cm,y=1cm]












\draw [line width=3pt] (0,2)-- (4.5,2);

\draw [line width=3pt] (4.5,2)-- (4.5,0);

\draw [line width=3pt] (4.5,0)-- (0,0);

\draw [line width=3pt] (0,0)-- (0,2);

\draw [line width=3pt] (6,2)-- (10.5,2);

\draw [line width=3pt] (10.5,2)-- (10.5,0);

\draw [line width=3pt] (10.5,0)-- (6,0);

\draw [line width=3pt] (6,0)-- (6,2);

\draw [line width=3pt] (12,2)-- (16.5,2);

\draw [line width=3pt] (16.5,2)-- (16.5,0);

\draw [line width=3pt] (16.5,0)-- (12,0);

\draw [line width=3pt] (12,0)-- (12,2);

\draw [line width=3pt] (18,2)-- (22.5,2);

\draw [line width=3pt] (22.5,2)-- (22.5,0);

\draw [line width=3pt] (22.5,0)-- (18,0);

\draw [line width=3pt] (18,0)-- (18,2);

\draw [line width=3pt] (24,2)-- (28.5,2);

\draw [line width=3pt] (28.5,2)-- (28.5,0);

\draw [line width=3pt] (28.5,0)-- (24,0);

\draw [line width=3pt] (24,0)-- (24,2);

\draw [line width=3pt] (3.375,8.5)-- (3.375,6.5);

\draw [line width=3pt] (3.375,6.5)-- (8.25,6.5);

\draw [line width=3pt] (8.25,6.5)-- (13.125,6.5);

\draw [line width=3pt] (13.125,6.5)-- (13.125,8.5);

\draw [line width=3pt] (13.125,8.5)-- (3.375,8.5);

\draw [line width=3pt] (15.375,8.5)-- (15.375,6.5);

\draw [line width=3pt] (15.375,6.5)-- (20.25,6.5);

\draw [line width=3pt] (20.25,6.5)-- (25.125,6.5);

\draw [line width=3pt] (25.125,6.5)-- (25.125,8.5);

\draw [line width=3pt] (25.125,8.5)-- (15.375,8.5);

\draw [line width=3pt] (1.125,13)-- (1.125,11);

\draw [line width=3pt] (1.125,11)-- (14.25,11);

\draw [line width=3pt] (14.25,11)-- (27.375,11);

\draw [line width=3pt] (27.375,11)-- (27.375,13);

\draw [line width=3pt] (27.375,13)-- (1.125,13);

\draw [line width=3pt] (9.375,-4.5)-- (9.375,-6.5);

\draw [line width=3pt] (9.375,-6.5)-- (19.125,-6.5);

\draw [line width=3pt] (19.125,-6.5)-- (19.125,-4.5);

\draw [line width=3pt] (19.125,-4.5)-- (14.25,-4.5);

\draw [line width=3pt] (14.25,-4.5)-- (9.375,-4.5);

\draw [line width=3pt] (2.25,-7)-- (2.25,-9);

\draw [line width=3pt] (2.25,-9)-- (26.25,-9);

\draw [line width=3pt] (26.25,-9)-- (26.25,-7);

\draw [line width=3pt] (26.25,-7)-- (7.125,-7);

\draw [line width=3pt] (7.125,-7)-- (2.25,-7);

\draw [line width=3pt] (8.25,2)-- (8.25,6.5);

\draw [line width=3pt] (20.25,2)-- (20.25,6.5);

\draw [line width=3pt] (25.125,2)-- (25.125,6.5);

\draw [line width=3pt] (15.375,2)-- (15.375,6.5);

\draw [line width=3pt] (14.25,2)-- (14.25,11);

\draw [line width=3pt] (13.125,2)-- (13.125,6.5);

\draw [line width=3pt] (27.375,2)-- (27.375,11);

\draw [line width=3pt] (3.375,2)-- (3.375,6.5);

\draw [line width=3pt] (1.125,2)-- (1.125,11);

\draw [line width=3pt] (9.375,0)-- (9.375,-4.5);

\draw [line width=3pt] (14.25,-4.5)-- (14.25,0);

\draw [line width=3pt] (19.125,-4.5)-- (19.125,0);

\draw [line width=3pt] (2.25,-7)-- (2.25,0);

\draw [line width=3pt] (26.25,-7)-- (26.25,0);

\draw [line width=3pt] (7.125,0)-- (7.125,-7);

\begin{scope}[font={\Huge}, every node/.style={scale=2}]
\draw[color=black] (2.37648864597189,.9) node {\(x_1\)};

\draw[color=black] (8.393665955594528,.9) node {\(x_2\)};

\draw[color=black] (14.374595209135103,.9) node {\(x_3\)};

\draw[color=black] (20.391772518757744,.9) node {\(x_4\)};

\draw[color=black] (26.372701772298317,.9) node {\(x_5\)};

\draw[color=black] (8.25,7.5) node {\(x_1 \lor x_2 \lor x_3 \)};
\draw[color=black] (20.25 ,7.5) node {\(x_3 \lor x_4 \lor x_5\)};

\draw[color=black] (14 ,11.9) node {\(x_1 \lor x_3 \lor x_5\)};

\draw[color=black] (14.5,-5.5) node {\(\overline{x_2}\lor \overline{x_3} \lor \overline{x_4}\)};

\draw[color=black] (14.5,-8) node {\( \overline{x_1} \lor \overline{x_2} \lor \overline{x_5}\)};
\end{scope}

\end{tikzpicture}}
    \caption{An example of monotone planar 3SAT
    instances \((x_1 \lor x_2 \lor x_3) \land (x_1 \lor x_3 \lor x_5) \land (x_3 \lor x_4 \lor x_5) \land ( \overline{x_1} \lor \overline{x_2} \lor \overline{x_5}) \land (\overline{x_2}\lor \overline{x_3} \lor \overline{x_4})  \) and it's corresponding
    rectilinear graph drawing.
    }
    \label{fig:mprsat-example}
\end{figure}

The fundamental idea of the reduction is to ``trace''
the rectilinear graph drawing of an \mrpsat instance 
\(\phi\) with gadgets and link the variables together such that any MST consists of one large central branch, called the \emph{base}, which spans through all the variables and has a single branch spanning each clause.
When tracing the variable and clauses,
minimum distance between any plot is \(1/2\)
and clauses are separated by from the variables by
a gap of size \(5\) and from other clauses by a larger gap of size \(11\).
This allows us to compute a lower bound on
the cost of the MST given \(\phi\).
Each variable is traced with type-2 plots
so picking only the top corners or bottom corners
as the entrance and exits can achieve the minimum distance between plots.
We show that MST where only the minimum distance between plots in variable gadgets achieves the lower bound and corresponds to a truth assignments satisfying \(\phi\).

We start by addressing any concerns related
to connecting plots across distances
during the tracing process and then focus on description of gadget used in the trace.
Lastly, we describe the necessary rescaling of rectilinear graph embedding to avoid any size issues that
could appear when tracing in gadgets.
In~\cref{fig:plot-notation}, we describe the meaning of various elements found in figures in the remaining sections.



\subsection{Connector Gadget}
The primary tool we use to connect plot in different gadgets is a \emph{Connector gadget} or
simply a \emph{connector}.
In a construction of a connector, we need to compute square roots of \(n\)-digit binary numbers on a Turing machine, where \(n\) is some number we specify.
It is clear two \(n\)-digit numbers can be multiplied on a Turing machine in polynomial time with respect to \(n\).
For this reduction we do not worry ourself any
further with the specifics of this runtime.
With this in mind, if we are given a \(n\)-digit number, the square
root of the number can be computing to \(\tau\) digits of accuracy
using binary search over all (\(n + \tau\))-digit numbers which require
at most \(\log(2^{n + \tau}) = O(n+\tau)\) number of multiplication calls.


\begin{figure}[t]
    \centering

    \begin{subfigure}[t]{0.45\textwidth}
        \centering
        \includegraphics[width=\textwidth]{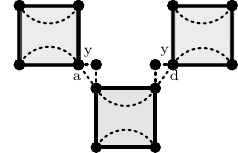}
        \label{fig:horizontal_connector}
    \end{subfigure}
    \hfill
    \begin{subfigure}[t]{0.4\textwidth}
        \centering
        \includegraphics[width=\textwidth]{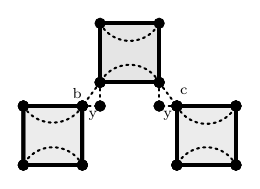}
        \label{fig:horizontal_c_mirror}
    \end{subfigure}

    \caption{Horizontal connector gadget connects two type-2 plots which is located on the same horizontal coordinates to their \(a\) to \(b\) corners with small error for \(y \in [0, 0.5]\).The right hand side of the figure shows the mirror configuration that allows us to connect the two plots via the bottom edges in the MST.}
    \label{fig:horizontal_connector_gadgets}
\end{figure}

\begin{figure}[t]
    \centering

    \begin{subfigure}[t]{0.3\textwidth}
        \centering
        \includegraphics[width=\textwidth]{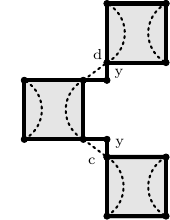}
        \label{fig:vertical_connector}
    \end{subfigure}
    \hfill
    \begin{subfigure}[t]{0.3\textwidth}
        \centering
        \includegraphics[width=\textwidth]{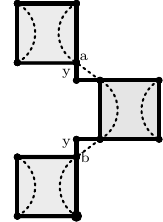}
        \label{fig:vertical_c_mirror}
    \end{subfigure}

    \caption{Vertical connector gadget connects two type-3 plots which is located on the same vertical coordinates to their \(c\) to \(d\) corners with small error for \(y \in [0, 0.5]\). The right hand side of the figure shows the mirror configuration that allows to connect the two plots via the right edges in the MST.}
    \label{fig:vertical_connector_gadgets}
\end{figure}

Given two plots of both type-2 or type-3 that lie on the same horizontal (or vertical) coordinates, a connector gadget is series of plots that connects an entrance and exit of each plot using edges of cost approximately \(1/2\).
Consider the case where we wish to connect
plots with the same horizontal coordinates.
As seen in~\cref{fig:2_3_components},
there is an easy arrangement of plots to connect an exit and entrance that are 2 or 3 units apart.
As every integer greater than \(1\) can be written as a sum of only twos and threes, this provides
a means to connects the plots when distance
is an integer and greater than 2.
When the plots instead have a non-integer distance, the last arrangement in~\cref{fig:2_3_components} is
able to span the fractional part of the distance.
This again can only occur if the distance between the plots is sufficiently large.

In all the shown arrangements, the optimal strategy is to route the connections through the bottom entrance/exits in the top plot and
then move down to top entrances/exits in the subsequent plot, giving a total cost of \(1\). The same reasoning applies to the vertical connector. Furthermore, the configuration can be mirrored to symmetrically connect the right side of the type-3 plots. In~\cref{fig:horizontal_connector_gadgets} and \cref{fig:vertical_connector_gadgets}, an arrangement is provide for every case.





We now show the following claim:
\begin{lemma}
    Let \(P_1\) and \(P_2\) be two type-2 (type-3) plots whose top and bottom (left/right) edges lie on a common horizontal (vertical) line. If the distance \(g\) between these aligned edges satisfies \(g \ge 4\), then there exists a sequence \(C\) of horizontal (vertical) connector units such that the set of plots formed by \(P_1\), \(P_2\), and \(C\) admits a minimum spanning tree in which all but at most two edges have length exactly \(\frac{1}{2}\). The remaining edges may deviate from \(\frac{1}{2}\) by at most \(\epsilon\), where \(\epsilon=\frac{1}{2^\tau}\) represents the maximum numerical error. Furthermore, the number of connector units required is \(O(\lceil g \rceil)\).
    \label{lemma:scaling_with_connectors}
\end{lemma}

\begin{proof}

Let \(g \in \mathbb{R}\) with \(g \ge 4\). Note that \(g\) can be expressed as a sum of segments of lengths \(2\), \(3\), and a remainder \(d \in [1,2]\). Formally, there exist integers \(a, b \ge 0\) such that \(g = 2a + 3b + d.\) Consider two type-2 plots \(P_1\) and \(P_2\) whose top edges lie on the same line parallel to the horizontal axis, where \(g\) is the horizontal distance between them. 
Using the connector gadgets shown in \cref{fig:2_3_components}, we may realize distances of length \(2\), \(3\), and any \(d \in [1,2]\). 
Thus, we place \(a\) connectors of length \(2\), \(b\) connectors of length \(3\), and one connector of length \(d\) to fill the gap between \(P_1\) and \(P_2\). Let \(C\) denote this collection of connectors.

Next, consider the minimum spanning tree on the set of plots \(C \cup \{P_1, P_2\}\).
As shown in \cref{fig:2_3_MST}, all elements can be joined using edges of length \(\tfrac{1}{2}\). 
Since connectors of length \(2\) and \(3\) introduce no numerical error, only the connector corresponding to the remainder \(d\) may have numerical error. As at most two such edges occur in the MST, the accumulated error is bounded by \(2\epsilon\).

Finally, the number of connector components required is \(a + b + 1 \le \lceil g \rceil = O(\lceil g \rceil).
\) An identical argument applies to bottom-edge alignment and to vertically aligned type-3 plots.

    \label{proof:scaling_with_connectors}
\end{proof}

\begin{figure}[t]
    \centering
    \includegraphics[width=0.8\textwidth]{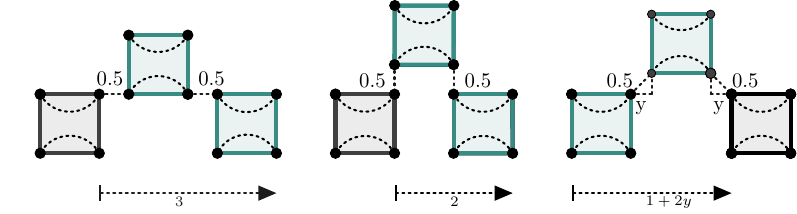}
	\caption{Horizontal components that can introduce lengths \(2,3\) or \(d \in [1,2]\). }
\label{fig:2_3_components}
\end{figure}

\begin{figure}[t]
    \centering
    \includegraphics[width=0.8\textwidth]{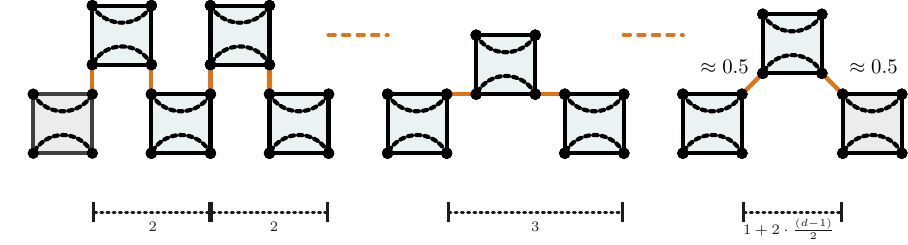}
	\caption{Grey plots represent the \(P_1\) and \(P_2\). Light green plots represent the added horizontal plots for connecting \(P_1\) and \(P_2\). Each orange edge represents \(0.5\) distance edges except for the last two edges that was used in the last horizontal connector. The last two values approximate \(0.5\) some error \(\epsilon\).}
\label{fig:2_3_MST}
\end{figure}



\subsection{Tracing the Embedding}
We now describe the tracing process which constructs an instance of the
\cpprob.
Although we are ignoring scale issues for the time being,
at this point we assume the instance has already been
rescaled to fall so the ends of each line and corner of every rectangle has integer coordinates.
It should be obvious that this rescaling can be done so every integer used in a coordinate is \(O(nm)\).

The tracing process starts by constructing the base.
Here first every variable in the embedding is replaced
with a variable  gadget as described in \cref{sec:trace-var}.
After the variable gadgets are added, they are then connected with link subgadgets.

Then, moving away from the base, each clause is traced as described in~\cref{sec:trace-clause}.
This starts with the placement of a gap gadget where the leg of clause
meets the variable in the embedding (see \cref{sec:trace-gap}).
At the other end of the leg, where it meets with clause,
a difference in gadgets is placed which is depending on which variable in the clauses we are current tracing.
In order from left to right: the first variable adds a left turn gadget,
the second variable adds a center gadget, the third variable produces a right turn gadget.
These gadgets are then connected with their corresponding gap
gadget using a vertical connector.
Finally all three gadgets placed on the clause are connected
with horizontal connectors.

During trace, we produce connector with
\(\tau ={2(5m+n-1)+7}\) so that the absolute
error is at most \(\epsilon = 1/2^{2(5m+n-1)+7}\).
This level precision in minor detail is required for our proof of correctness later and all arguments still
hold with only minor modification if exact values
for square root computation are used instead.

\begin{figure}[t]
    \centering
    \includegraphics[width=0.8\textwidth]{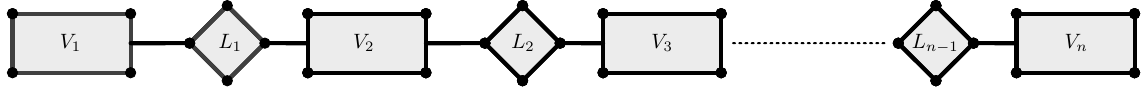}
	\caption{Base contains variable gadgets connected via link gadgets.}
\label{fig:base}
\end{figure}

\begin{figure}[t]
    \centering
    \includegraphics[width=0.55\textwidth]{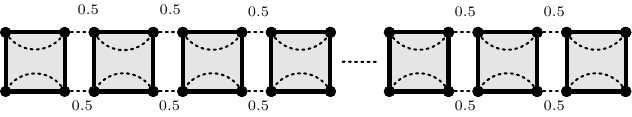}
	\caption{For each variable, we create a variable gadget.}
\label{fig:variable_subgadget}
\end{figure}

\begin{figure}[t]
    \centering
    \includegraphics[width=0.9\textwidth]{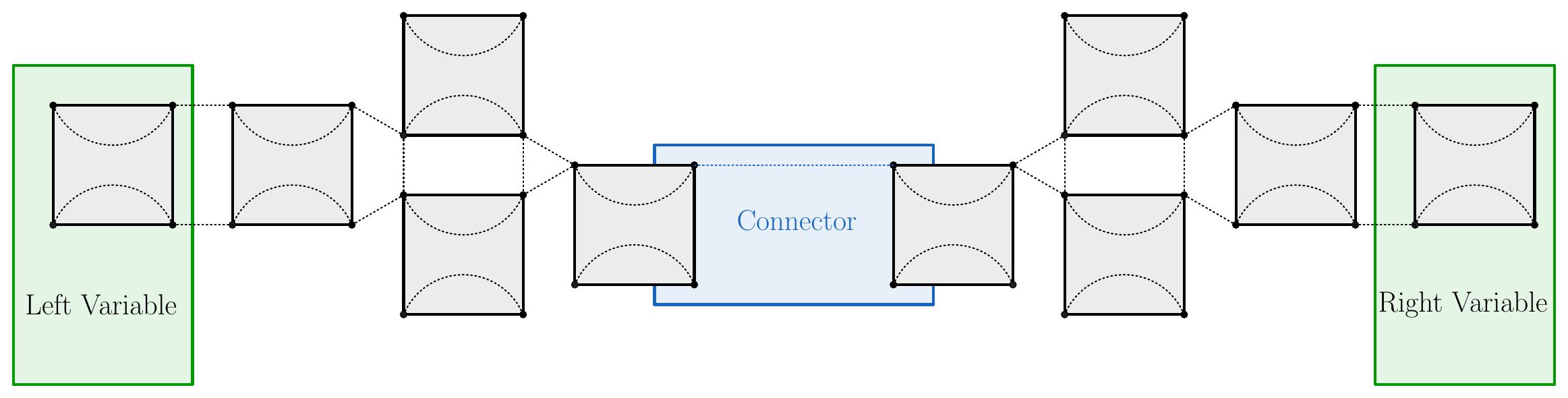}
    	\caption{
        The organization of the
        link gadget.
        Each dotted line represents a cost of size \(\frac{1}{2}.\)}
\label{fig:link_subgadget}
\end{figure}

\begin{figure}[t]
    \centering
    \includegraphics[width=0.9\textwidth]{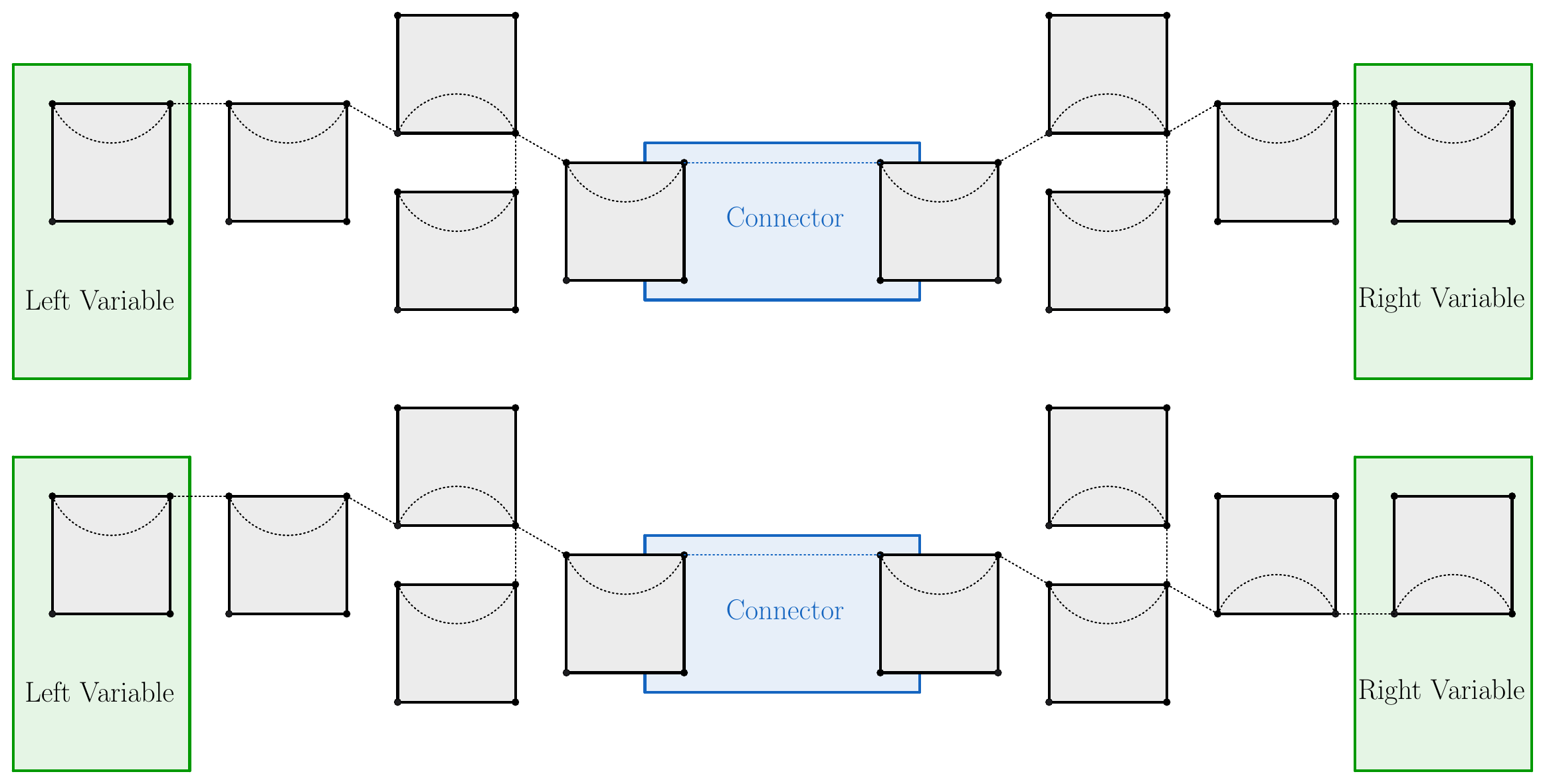}
    	\caption{The top image depict the tree through the link gadget that connects a variable assigned \emph{true} to another variable assigned \emph{true}. The bottom depicts a tree that connects a variable assigned \emph{true} to a variable assigned \emph{false}.}
\label{fig:link_routes}
\end{figure}

\subsection{Trace of the \emph{Base}}
\label{sec:trace-base}

The \emph{base} forms the central structure of the reduction instance. It contains the variable gadgets corresponding to all Boolean variables in \(\phi\), connected by special components called \emph{link subgadgets}. Together, these variable gadgets and link subgadgets form a spine-like structure that we refer to as the base. In the following sections, we describe each type of gadget in detail. Please refer to Figure~\ref{fig:base}.

\subsubsection{Variable Gadget}
\label{sec:trace-var}

For each boolean variable $x_i$ in \(\phi\), we create a variable gadget. The variable gadget is simply a row of $O(\ell_i)$ type-2 plots where the distance between adjacent plots is \(\frac{1}{2}\) and $\ell_i$ is the number of times $x_i$ and $\overline{x_i}$ appearing in $\phi$. 
Due to the initial scaling of the embedding,
    each plot in a variable gadgets starts at
    a multiple of \(3/2\) along the horizontal axis.

Note that when we connect these plots using an MST where either the top corners of the plots or the bottom corners of the plots will be connected. In fact, this is used as the truth assignment for that particular variable. Please refer the Figure \ref{fig:variable_subgadget}.
We refer to the top entrances and exits of
the variables as positive and the bottom as
negative.

\subsubsection{Link Gadget}
\label{sec:trace-link}

The link gadget connects adjacent variable gadgets during tracing and allows for the variables assignment to switch without taking edges of length greater than \(1/2\). Whenever we have two adjacent variable gadgets, we need to ensure that these two variable gadgets can switch its truth assignment, in order to facilitate that we use a special gadget which we refer to as link gadget. Please refer to \cref{fig:link_subgadget}. Suppose we want to connect variables gadget \(v_{i-1}\) with \(v_{i}\), then the last plot of \(v_{i-1}\) and the first plot of the \(v_{i}\), we add this special link gadget in between these two plots. Note that link gadget contains 3 parts. Left rung which connects to \(v_{i-1}\), middle part which is a horizontal connector and the third part that connects the horizontal connector to the right rung, where right rung connects to the first plot in \(v_{i}\). This gadget allows to switch the truth assignment between variable gadgets if it requires to do so as seen~\cref{fig:link_routes}.

\subsection{Trace of the Clauses}
\label{sec:trace-clause}

The trace of clauses consists of 3 parts:
a turn gadget, a clause center gadget,
and a gap gadget to connect the clause with
it's variables.
The layout of the trace of a clause is shown in
Figure~\ref{fig:total-clause-trace}.

\begin{figure}
    \centering
    \resizebox{.7\linewidth}{!}{\input{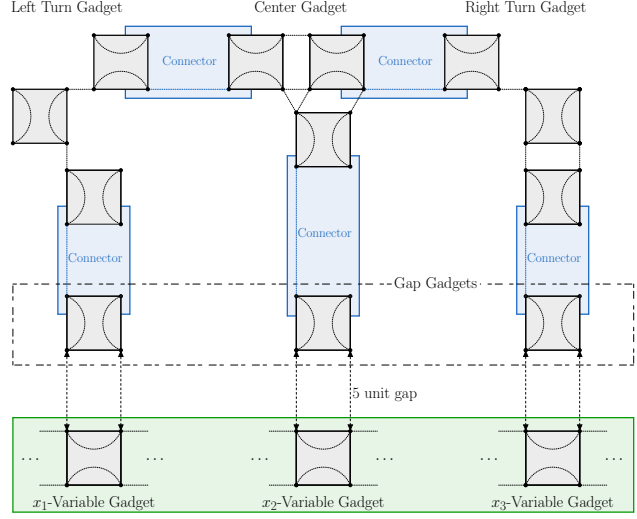}}
    \caption{
    The green box at the bottom denotes elements produced
    when tracing the base.
    The dotted lines inside the connector represent entrance/exits pairs connected by the gadget. Dotted lines outside of the connector represent distance of 1/2 units unless otherwise specified.
    }
    \label{fig:total-clause-trace}
\end{figure}
\begin{figure}[t]
    \centering
    \includegraphics[width=0.9\textwidth]{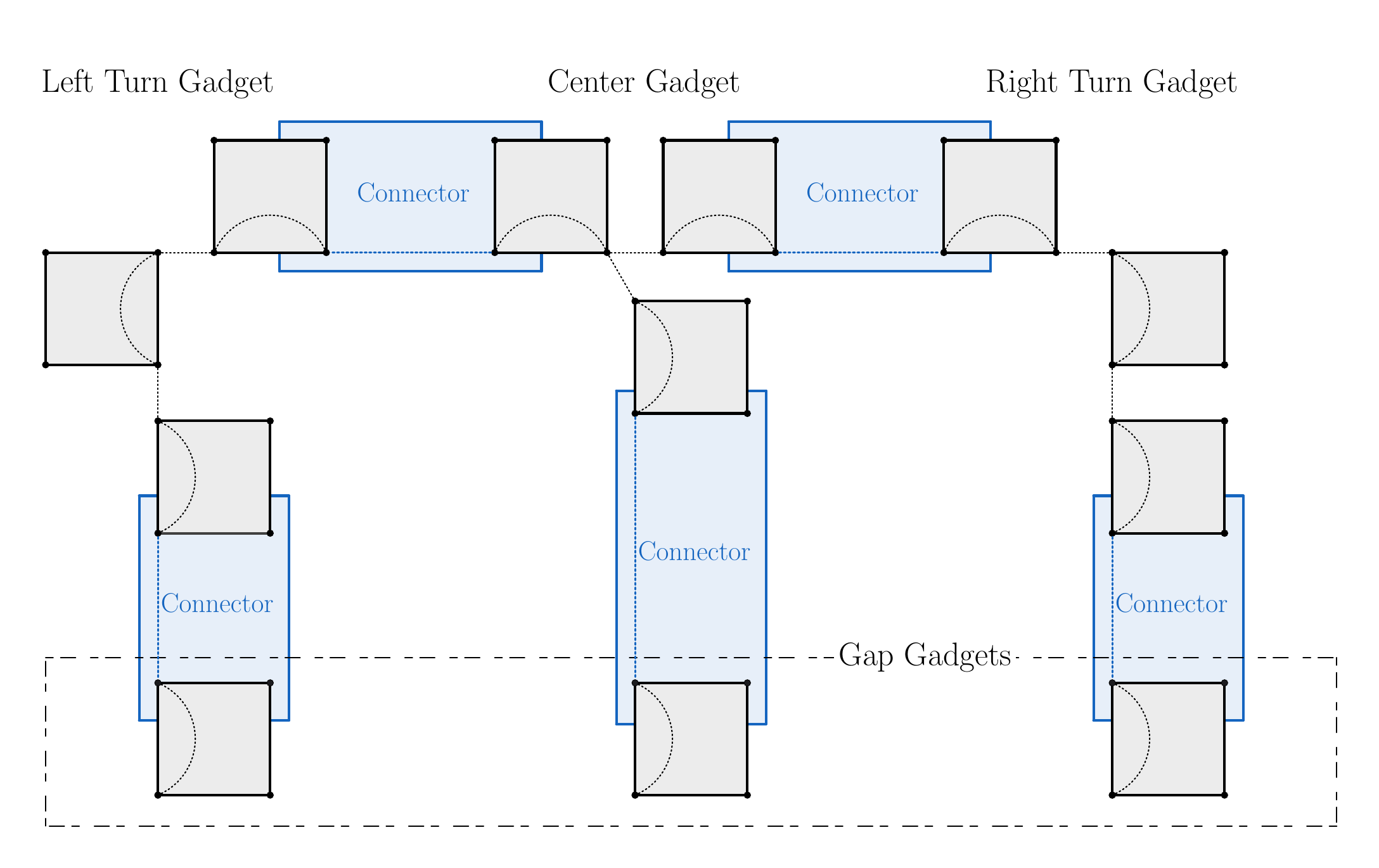}
    	\caption{An example of a MST for plots produced during the trace of the clause taking only edges that are cost \(1/2\) or edges in a connector.
        Each dotted line represents a the shortest path of size \(1/2.\)}
\label{fig:clause-routing}
\end{figure}

\subsubsection{Gap Gadgets}
\label{sec:trace-gap}
The gap gadget is placed wherever the
\emph{leg} of the clause touches a variable in the embedding.
In the case of clause with only positive literals, the bottom-left corner of the gap gadget is placed 5 units above the where lines in the embedding meet.

Again as we assume the embedding has already been scaled to lie on
integer grid with unit in grid corresponding to 3/2 units in the
reduction, we have gap gadget placed directly above a plot in
the variable gadget as seen in~\cref{fig:total-clause-trace}.

\subsubsection{Turn and Centers Gadget}
\label{sec:trace-turn-and-padding}

The turn gadgets are placed where the
left and right legs of the embedding connect with the clause.
The top-left corner of the bounding box for the left turn gadget is
placed where left leg meets the clause in the embedding.
The top-right corner of the bounding box for the right turn gadget is placed
where the right leg meets the clause in the embedding.
The top-left corner of the center gadget is
placed 5/4 units left of where the center leg meet the clause.
This ensures the bottom plot in the center gadget to align with the gap
gadget for the leg as shown in~\cref{fig:total-clause-trace}.

We use connector to connect the turns to the center as well as the
turns and centers to the gap gadgets.
As shown in \cref{fig:clause-routing}, an important property
of the center gadget is it allows for a spanning tree 
using only edges of only cost 1/2 to exit over plots in
tracing of the clause.
Note that if we select a different gap gadget as a start,
we can still produce a spanning tree with a similar topology
due to the triangle in the center gadget.
\subsection{Scaling and Padding}

We describe all the required scaling of the rectilinear embedding to ensure we always
have space for the gadgets and topology of MST to behave as described earlier.
As mentioned before, we assume that the rectilinear embedding has already been transformed to
fall on an integer grid, with the scale being 3/2 in both horizontal and vertical directions.
Finally, the entire drawing is scaled by an integer multiplier until
all the following are satisfied:
\begin{itemize}
        \item The minimum distance between legs of the clauses
          is greater than 15 units apart.
          As the total size of a connector and any two gadget is placed on the clause is not more than 10,
          this provide sufficient room to place
          any components needed between the legs.
          Additional 5 units are added to 10 for padding reasons.

    \item The minimum distance between clause is greater than 14.5 units a part.
          Note that the left/right turn gadgets extend down 3.5 units from line of clause and are the largest components on the line,
          thus 3.5 units require for us to construct these gadgets.
          An additional 11 units are added to 3.5 units to act as padding.

    \item The minimum distance between clauses
           and
          the minimum size of a clause is greater than 14.5 units.
    \item No clause or variable is closer than 13.5 units apart.
           We require 6.5 units for the gap gadget, 4 units for a connector
          and at most 3 units for any gadget place on clause directly.
    \item The minimum distance between variables in the embedding
          is 13. As the width of each side of the link gadget is no more than 4.5
          units and a connector requires 4 units, this provides enough
          room for the entire link gadget as well as the connector.

\end{itemize}

Now let the reduction instance for 
\(\phi\) be denoted by \(P_\phi\).
Denote the plots used to trace the \(i\)-th
clause by \(C_i\).
\begin{proposition}[*]
    \label{lem:clause-dist}
    Let \(C_i\) and \(C_j\) be the plots used
    to trace two different clauses in \(\phi\).
    Then for any assignment of entrances and
    exits for \(C_i\) and \(C_j\)
    \[11 \leq  \min_{A \in C_i, B \in C_j} d_{cp}(A, B).\]
\end{proposition}

\begin{proposition}
    \label{lem:prop:not-clause-dist}
    Let \(C_i\) be the plots used
    to trace some clause in \(\phi\).
    Then for any assignment of entrances and
    exits for \(C_i\) and \(C_j\)
    \[ 5 \leq  \min_{A \in C_i, B \in P_{\phi} \setminus C_i} d_{cp}(A, B).\]
\end{proposition}
\begin{proof}
By \cref{lem:clause-dist}, we need to only
verify that no plot on the base is closer than \(6\) units to any plot on the clause. It is clear that the closest
plots between the base and the plot used to trace a clause are the single plot in gap gadget and the plots that fall directly under the gap gadgets. As the minimum distance between clause and this plot is \(5\) the result holds.
\end{proof}

\subsection{Hardness Proof}
We now go about showing that any MST for \(\phi\)
has the topology claimed earlier before showing the if and only if claim.

\begin{proposition}[*]
    Let \(C_i\) be the set of plots used in the
    clause gadget for the \(i\)th clause.
    In any MST \(\mathcal{T}\) of \(P_{\phi}\), we have \(\mathcal{T}[C_i]\) is a connected component. Further \(\mathcal{T}[C_i]\) connects 
    to the remainder of \(\mathcal{T}\) with exactly one edge to the base that occurs between a gap gadget in the clause and plot in the vertex gadget.
    \label{prop:topology}
\end{proposition}

We now define two types of edges in any MST for \(P_\phi\).
A \emph{good} edge is an edge that is of length 1/2,
an edge in a connector of length within \(\epsilon\) of 1/2, or edge length of 5 between the gap gadget and variable gadget.
Any edge which is not a good edge is said to be a \emph{bad} edge.
This provide the following claim:

\begin{lemma}
    \label{lem:good-is-sat}
    Any MST of \(P_\phi\) which uses only good edges
    corresponds to a truth assignment that satisfies \(\phi\).
\end{lemma}
\begin{proof}
    Note as only good edges are chosen, only
    the positive or negative side of each variable 
    corresponding to a true or false assignment
    of variable.
    Further each clause must connect to the base once by~\cref{prop:topology}.
    Therefore, if the clause is positive, then some variable in the clause must be assigned true where gap gadget and variable gadget
    have edge of cost 5.
    The argument is identical if the clause is negative.
    Thus, the truth assignment produced by the MST must satisfy \(\phi\).
\end{proof}

We now show the following lemma necessary for
our reduction.

\begin{lemma}[*]
    \label{lem:crop-iff}
    The \mrpsat instance \(\phi\) is satisfiable if and only if the MST of plot of \(P_\phi\) has a cost less than \(\frac{9}{2}\cdot m + \frac{1}{2}(|P_\phi| -1) + 1/128\). 
\end{lemma}

\begin{theorem}
\label{cpmst-hard}
The decision problem variant of \cpprob 
is NP-hard.
\end{theorem}
\begin{proof}
As we assume the initial rectilinear graph drawing falls on the integer
grid using only integers at most \(O(nm)\) in size,
after all rescaling, the largest real value in rectilinear graph drawing is
still \(O(nm)\). Thus whenever we draw a connector we spend no more
than \(poly(n,m)\) time. Tracing the clauses requires \(5\) connectors
and gadgets of fixed size, and tracing the bases is in the size of the grid and uses
\(n-1\) connectors, thus the entire construction of the reduction instance can be
done polynomial time.  In combination with~\cref{lem:crop-iff}, this theorem is proved.

\end{proof}

Theorem~\ref{cpmst-hard} shows that computing ${MST}_{cp}(P)$ is NP-hard. As ${MST}(P)$ is really ${MST}_{cp}(P)$ plus the additional path inside each plot (with a total internal length $C$), computing ${MST}(P)$ is also NP-hard. In the next section, we present a factor-4 approximation for ${MST}(P)$.

\section{A Simple Approximation}

We present a simple approximation algorithm for ${MST}(P)$, whose cost is $c({MST}(P))=c({MST}_{cp}(P))+C$.
The algorithm is as follows:
\begin{enumerate}
    \item Compute the Voronoi diagram of squares in $P$, using the algorithm by McAllister, Kirkpatrick and Snoeyink \cite{DBLP:journals/dcg/McAllisterKS96}. Based on the dual of this diagram, compute the minimum spanning tree
    of the squares in $P$, under the $L_2$ distance, $MST_2(P)$.
    \item For each of the edges in $MST_2(P)$, say $e=(s_i,s_j)$, arbitrarily select a pair of entry/exit points
    following the types of square $s_i$ and $s_j$ (if their entry/exit points are not yet fixed), and then compute the shortest geodesic distance between the selected entry/exit
    points in $s_i$ and $s_j$.
\end{enumerate}

\begin{theorem}[*]
Given a set $P$ of $n$ square plots, the minimum spanning tree of $P$ admits a factor-4 approximation which runs in $O(n^2\log n)$ time.
\end{theorem}

Using the standard trick to approximate TSP of $P$, i.e., doubling the edges of ${MST}(P)$, we can obtain a factor-8
approximation to the TSP problem of $P$ --- assuming entering/exiting a plot twice is allowed.

\section{Concluding Remarks}
In this paper, we show that the problem of computing the MST of a set of disjoint axis-aligned squares under the crop plot distance is NP-hard; moreover, we present a simple factor-4 approximation. For the approximation algorithm, one immediate question is what would happen if we could model the crop plots as axis-aligned rectangles, possibly with different size. The quick answer is that if the aspect ratio of each rectangle is bounded and if the ratio of the largest/smallest rectangle length (resp. height) is also bounded --- quite realistic in a farm, a constant factor still holds. 

For the approximation of the TSP of $P$, certainly the factor-8 approximation is quite straightforward --- if entering/exiting a plot twice is allowed. Even under this assumption, would it be possible to improve the factor? At the point, the general idea similar to that of Christofides \cite{Christofides1976}, which is to compute a minimum weight matching $M$ on the nodes on ${MST}$ with odd degrees then compute an Euler tour on ${MST} \cup {M}$, does not seem to work, as the crop plot distance
does not satisfy the triangle inequality. More research is needed along this line.


\section*{Appendix}

\subsection*{Proof of Proposition 1}
\begin{proof}
    First suppose \(A\) and \(B\) are both either in a gap gadget or
    connector from the gap gadget to a turn or center gadget.
    As every leg is 15 units apart and the vertical connector can be no more
    than 1.5 unit horizontally from the leg in the embedding
    we have that no plot can be closer than 12 units apart and we are done.

    Now without loss of generality, suppose \(A\) is either in the turn or clause gadget or the connector among these gadgets.
    Note we only need to worry about the case when \(A\) is to the left or right of clause. We first consider the case where \(A\) is to the left.
    Observe the left leg of the clause can be no more than 1 unit further from \(B\) than \(A\) and this only occurs when \(B\) is in the left turn gadget. Regardless, we have all entrance and exits of \(A\) and \(B\) are at least 11 units apart.
    The case where \(A\) is to the right of clause of \(B\) follows similarly.
    
    Finally, suppose \(A\) and \(B\) are both in either a turn or center or
    the connector among them. At there closest these plots are 11 units apart in the \(x\) coordinate due to scaling. Thus their distance must be greater than 11.
\qed
\end{proof}

\subsection*{Proof of Proposition 3}
\begin{proof}
    Suppose for the sake of contradiction \(\mathcal{T}[C_i]\) is not connected in some MST of \(P_\phi\),  \(\mathcal{T}\). Recall \(C_i\) has an MST of cost exactly \((|C_i|-1)/2\).
    Consider any two plots \(A, B\) that are adjacent in some MST for \(C_i\) but not in same connected component in \(\mathcal{T}[C_i]\). Note that the path in \(\mathcal{T}\) that connects these two plots cross two
    edges of length at least \(5\) by~\cref{lem:prop:not-clause-dist}.
    Connecting the entrance of \(A\) to the exit of \(B\) produces a cycle in \(\mathcal{T}\).
    Observe the length of this edge can not be greater than half of the boundary length of \(A\) and \(B\) plus 1/2,
    thus is at most 4.5, thus we can produce a smaller spanning tree for \(P_\phi\) by removing an edge of length 5. 
    This contradicts \(\mathcal{T}\) being an MST for \(P_\phi\).
    Thus \(\mathcal{T}[C_i]\) must be connected.

    Now observe that, regardless of points selected
    as entrance and exit on any gap gadget in the trace of clause, it can be connected to the plot in the variable gadget directly below it with a path of length no greater than half of the length to the boundary plus the size of gap, which is 5.
    As this path has a total length of 10 and by \cref{lem:clause-dist} any connection to any other clause has to have length greater than 11,
    an MST always connects to the base of \(P_\phi\) instead of another clause. 
\qed
\end{proof}

\subsection*{Proof of Lemma 3}
\begin{proof}
    First we prove the forward direction. Assume the \mrpsat instance \(\phi\) is satisfiable.
    We begin by constructing the base and connecting the variable gadgets. For each variable assigned \emph{true}, we select the top edges of the plots in its corresponding gadget and connect consecutive plots using edges of length \(\frac{1}{2}\). For variables assigned \emph{false}, we instead select the bottom edges and connect the plots in the same manner.
    Then for each clause in the \mrpsat instance, we pick one variable that can make the clause true and connect that variable to the clause leg. This would introduce additional \(\frac{9}{2}\) distance (\(5 - \frac{1}{2}\)). 
    Using this gap edge we can cover everything in the clause gadget. 
    Since we cover each plot, we introduce at least \(\frac{1}{2}\cdot (|P_{\phi}|-1)\) cost to the MST, additionally, since we have to use the gap edges, we introduce an additional \(\frac{9}{2}m\) cost to the overall MST.
    Moreover, numerical approximation contributes at most \(\frac{1}{128}\) error to the MST cost as there are \(5m+n-1\) connectors in \(P_{\phi}\) and \({2\cdot(5m+n-1)\epsilon}=\frac{1}{128}\frac{2\cdot(5m+n-1)}{2^{2\cdot(5m+n-1)}}\leq\frac{1}{128}\).
    This concludes the forward direction.

    Similar to the previous argument,
    we see the
    minimum score obtainable by an MST of \(P_\phi\), regardless if \(\phi\) is satisfiable or not,
    is \(\frac{9}{2}\cdot m + \frac{1}{2}(|P_\phi| -1) - 1/128\).
    This occurs when each clause is connect to
    base with edge of  cost \(5\), which by~\cref{lem:prop:not-clause-dist} is the smallest possible cost,
    and every other edge is of minimum cost.

    We now prove the backward direction. We prove this claim by contradiction. Suppose we have an MST of size \(\frac{9}{2}\cdot m + \frac{1}{2}(|P_\phi| -1) + 1/128\) and \(\phi\) is not satisfiable.
    As \(\phi\) is not satisfiable, \cref{lem:good-is-sat} tells us that a bad edge must occur in the MST.
     If the bad edge appears in place of an edge of minimum distance or in a connector, then it adds more than \(\frac{1}{4}\) additional cost. 
     On the other hand, if the bad edge appears in place of an edge of cost \(5\), then
     it incurs an additional \(\frac{\sqrt{101}}{2}-5 > \frac{1}{64}\) cost.
     In either case, we increase the cost of the minimum solution by more than \(1/64\).
     Therefore, the cost of the MST must be greater than \(\frac{9}{2}\cdot m + \frac{1}{2}(|P_\phi| -1) + 1/128\).
     This is a contradiction. Therefore, the backward direction is proven. 
     
\qed
\end{proof}

\subsection*{Proof of Theorem 2}

\begin{proof}
First note that $c(MST_{cp}(P))\geq c(MST_2(P))$. For each edge $e=(s_i,s_j)$ in $MST_2(P)$ (with $s_i,s_j\in P$), let the $L_2$-distance between two squares 
$s_i$ and $s_j$ be $d(s_i,s_j)$. Then geodesic distance returned by the algorithm $d'_{cp}(s_i,s_j)$ satisfies
$d'_{cp}(s_i,s_j)\leq 3+d(s_i,s_j)$, this can be achieved by connecting the entry/exit points in $s_i$ and $s_j$ via
walking around the corners in $s_i,s_j$ only. Consequently, the total approximation cost ${APP}$ is
\[{APP}\leq 3(n-1)+c(MST_2(P))+C;\]
while the optimal solution ${OPT}$ satisfies
\[{OPT}\geq c(MST_{cp}(P))+C\geq c(MST_2(P))+C.\]
Consequently, the approximation factor is 
\[\frac{{APP}}{{OPT}}\leq 1+\frac{3(n-1)}{c(MST_2(P))+C}\leq 4,\]
as $c(MST_2(P))\geq 0$ and $C\geq n$.

As the algorithm by McAlister, Kirkpatrick and Snoeyink only takes $O(n\log n)$ time (when the input is the set $S$ of $n$ squares), the total running time of the algorithm is dominated by computing the geodesic distance between vertices in $s_i$ and $s_j$ (for all edges $e=(s_i,s_j)$ of $MST_2(P)$),
which could certainly be computed using Dijkstra's algorithm on the visibility graph of the vertices of squares in $P$. However,
with the algorithm by Guo, Maheshwari and Sack \cite{DBLP:conf/aaim/GuoMS08}, we could build a data structure in $O(n^2\log n)$ time and $O(n^2)$ space such that each geodesic distance query can be done in $O(n\log n)$ time. This gives us a total running time of $O(n^2\log n)$.
\qed
\end{proof}
\end{document}